\documentclass[11pt]{amsart}

\date{\today}

\usepackage{etoolbox}
\makeatletter
\patchcmd{\@maketitle}
  {\ifx\@empty\@dedicatory}
  {\ifx\@empty\@date \else {\vskip3ex \centering\footnotesize\@date\par\vskip-4ex}\fi
   \ifx\@empty\@dedicatory}
  {}{}
\patchcmd{\@adminfootnotes}
  {\ifx\@empty\@date\else \@footnotetext{\@setdate}\fi}
  {}{}{}
\makeatother

\usepackage{amsmath,amssymb,amsthm}
\usepackage{mathtools}
\usepackage{graphicx}
\usepackage[labelfont={bf},margin=1cm]{caption}
\usepackage{subcaption}
\usepackage{fullpage}
\usepackage[dvipsnames,table]{xcolor}
\usepackage[colorlinks=true,citecolor=RoyalBlue,linkcolor=Red]{hyperref}
\usepackage{tikz}
\usetikzlibrary{snakes}

\usepackage{array}
\newcolumntype{C}[1]{>{\centering\arraybackslash}m{#1}}

\usepackage[sc]{mathpazo}
\usepackage[T1]{fontenc}

\usepackage[capitalize,nameinlink,noabbrev,nosort]{cleveref}

\newtheorem{thm}{Theorem}
\newtheorem{lem}[thm]{Lemma}

\crefname{lem}{Lemma}{Lemmas}
\crefname{thm}{Theorem}{Theorems}

\newcommand{\x}{\mathbf{x}}
\newcommand{\y}{\mathbf{y}}
\renewcommand{\L}{\mathbb{L}}
\newcommand{\Lp}{\mathbb{L}^+}
\newcommand{\z}{\zeta}
\newcommand{\oz}{\overline{\zeta}}
\newcommand{\e}{\eta}
\renewcommand{\oe}{\overline{\eta}}

\newcommand{\btm}{\vspace{1ex}\(\displaystyle} 
\newcommand{\etm}{\)\vspace{1ex}} 

\usepackage[backend=bibtex,maxbibnames=99,sorting=nyt]{biblatex}
\renewbibmacro{in:}{}
\DeclareFieldFormat[misc]{title}{\mkbibquote{#1}}

\DeclareFieldFormat[article]{volume}{\mkbibbold{#1}}
\DeclareFieldFormat{url}{\href{#1}{URL}}
\DeclareFieldFormat{doi}{\href{https://doi.org/#1}{DOI}}
\DeclareFieldFormat{eprint:arxiv}{arXiv:\href{https://arxiv.org/abs/#1}{#1}}

\title{Adsorbing staircase polygons subject to a force}
\author{Nicholas R. Beaton}
\address{School of Mathematics and Statistics, The University of Melbourne, VIC 3010, Australia}

\begin{document}
\maketitle

\begin{abstract}
We study several models of staircase polygons on the $45^\circ$ rotated square lattice, which interact with an impenetrable surface while also being pushed towards or pulled away from the surface by a force. The surface interaction is governed by a fugacity $a$ and the force by a fugacity $y$. Staircase polygons are simplifications of more general self-avoiding polygons, a well-studied model of interacting ring polymers. For this simplified case we are able to exactly determine the limiting free energy in the full $a$-$y$ plane, and demonstrate that staircase polygons exhibit four different phases, including a ``mixed'' adsorbed-ballistic phase.
\end{abstract}

\section{Introduction}\label{sec:intro}

The study of lattice walks as models of polymers in solution has a rich history -- see~\cite{vanRensburg2000Statistical} for a thorough treatment. One particular scenario that has yielded many results is that of a polymer adsorbing onto an impenetrable surface. This can be elegantly modelled by a self-avoiding walk in a half-space of a lattice, with a temperature-dependent fugacity (Boltzmann weight) associated with points at which the walk touches the boundary of the half-space. A variety of exact~\cite{Beaton2014Critical2,Hammersley1982Selfavoiding} and numerical~\cite{Beaton2012Twodimensional,Guttmann2014Pulling,vanRensburg2004Multiple,Luo2008Critical} results have been found, particularly regarding the location and nature of the adsorption phase transition.

A related scenario is that of a long polymer tethered to a surface, with a pulling force applied to part of polymer at some angle to the surface. Atomic force microscopy~\cite{Haupt_1999} has allowed for such experiments in the laboratory. Lattice models are again a useful tool for investigating such systems~\cite{Beaton_2015,Guttmann2014Pulling,Janse_van_Rensburg_2017,Krawczyk2005Pulling}. Many of these and other studies have considered polymer models incorporating both interactions with the surface and a pulling force, and the phase diagrams that result from the interplay between the two.

Instead of linear polymers, in nature there also exist ring polymers, which have a natural circular structure. The relevant lattice objects are then self-avoiding polygons, rather than walks. It is again possible to investigate adsorption~\cite{vanRensburg1999Collapsing,MR1953336,Soteros_1992}, pulling forces~\cite{Atapour2009Stretched,Beaton_2016,vanRensburg2008Knotting}, and both~\cite{GvRJW2017Polygons}.

With some exceptions (e.g.~\cite{Beaton_2015,Beaton2014Critical2}), most exact quantitative results regarding interacting self-avoiding walks and polygons are non-rigorous, and typically depend on assumptions about the scaling limit of the model. For this reason (and others), there has been considerable attention given to subclasses, which often allow for exact solutions for various thermodynamic quantities. Some of the simplest models which still display rich physical behaviour include directed paths (see~\cite{vanRensburg2003Statistical} for a thorough review) and polygons~\cite{vanRensburg1999Adsorbing}.

\subsection{Interacting self-avoiding walks and polygons}\label{ssec:saws_and_saps}

In the most general model, we let $p_n(v,h)$ be the number of self-avoiding polygons of length $n$ on $\mathbb Z^d$, counted up to translation, with
\begin{itemize}
\item all vertices having nonnegative $\x_d$-coordinate,
\item $v>0$ vertices lying in the hyperplane $\x_d=0$, and
\item $h$ being the maximal $\x_d$-coordinate among all vertices.
\end{itemize}
(Note that $p_n(v,h) \geq 0$ only if $n$ is even.) Then define the partition function
\[P_n(a,y) = \sum_{v,h} p_n(v,h) a^v y^h.\]
Here we interpret $a=\exp(-\alpha/kT)$ and $y=\exp(f/kT)$, where $\alpha$ is the energy associated with a single surface contact, $f$ is the pulling force, $T$ is absolute temperature, and $k$ is Boltzmann's constant. When $\alpha<0$, adsorption may occur for sufficiently small $T$. Likewise, for sufficiently large $f$, the polygons may become ballistic, and reach vertices at distance $O(n)$ above the surface.

For all $d\geq 2$ and $a,y\geq 0$, it is known~\cite{vanRensburg2008Knotting,Soteros_1992} that the limiting free energies
\[\kappa_0(a) := \lim_{n\to\infty} \frac{1}{2n}\log P_{2n}(a,1) \qquad \text{and} \qquad \lambda_0(y) := \lim_{n\to\infty} \frac{1}{2n}\log P_{2n}(1,y)\]
exist. They are convex functions of $\log a$ (resp.~$\log y$), and are thus continuous and almost-everywhere differentiable. It is also known~\cite{Soteros_1992} that $\kappa_0(a)$ has a point of non-analyticity at some $a^0_c>1$, separating the desorbed (free) phase $a<a^0_c$ (where $\kappa_0(a) = \kappa$, the connective constant of the lattice) and the adsorbed phase $a>a^0_c$ (where $\kappa_0(a) > \kappa$). Likewise, $\lambda_0(y)$ has a point of non-analyticity at $y^0_c=1$~\cite{Beaton_2015,GvRJW2017Polygons,vanRensburg2008Knotting}, separating the free phase (where $\lambda_0(y) = \kappa$) and the ballistic phase (where $\lambda_0(y) > \kappa$).

In $d\geq 3$ dimensions, it has also been shown~\cite{GvRJW2017Polygons} that the two-parameter free energy
\[\psi_0(a,y) := \lim_{n\to\infty}\frac{1}{2n} \log P_n(a,y)\]
exists for all $a,y>0$. Moreover
\begin{equation}\label{eqn:3d_saps_fe}
\psi_0(a,y) = \max\left\{\kappa_0(a), \lambda_0(y)\right\}.	
\end{equation}
The relation~\eqref{eqn:3d_saps_fe} completely (qualitatively) characterises the phase diagram for $d\geq3$, showing that there are three phases: free (when $a<a^0_c$ and $y<1$), adsorbed (when $a>a^0_c$ and $\kappa_0(a) > \lambda_0(y)$) and ballistic (when $y>1$ and $\lambda_0(y) > \kappa_0(a)$). The model has two natural order parameters: the limiting density of surface visits
\[\mathbf V := \lim_{n\to\infty} \left\langle \frac{v}{n}\right\rangle = a\frac{\partial}{\partial a} \psi_0(a,y),\]
where the expectation is taken with respect to the Boltzmann distribution on polygons of length $n$; and the limiting ``density'' of the height
\[\mathbf H := \lim_{n\to\infty} \left\langle \frac{h}{n} \right\rangle = y\frac{\partial}{\partial y} \psi_0(a,y).\]
Then in the free phase $\mathbf V = \mathbf H = 0$, while in the adsorbed phase $\mathbf V>0$ and in the ballistic phase $\mathbf H >0$.

In $d=2$ dimensions less is known, with only bounds on the $\liminf$ and $\limsup$ having been proven. In particular, it is not known whether there are only three phases as in $d\geq 3$, or whether there are one or more additional, ``mixed'' phases, where the free energy depends on both $a$ and $y$ and in which the polygons are both adsorbed and ballistic.

\subsection{Outline of the paper}\label{ssec:outline}

In this paper we investigate a directed version of the adsorbing and pulled/pushed polygons discussed above. In particular, our goal is to find and study a simple two-dimensional model which displays a mixed adsorbed and ballistic phase. To this end we employ staircase polygons, for which a number of existing enumerative results are already known.

In \cref{sec:themodel} we define interacting staircase polygons and two main subclasses which will be of use in our investigation, and also state our main theorems. In \cref{sec:pairs-of-paths} we make use of well-known results which relate staircase polygons to pairs of nonintersecting paths, in order to exactly enumerate one of our subclasses. The asymptotic behaviour of these enumerations is covered in \cref{sec:grafted_asymps}, leading to a complete phase diagram. In \cref{sec:centred_lbs,sec:ub} we find lower and upper bounds respectively for all staircase polygons, with one proof for the lower bound being deferred to \cref{app:bridges}. In \cref{sec:conclusion} we offer some closing remarks.

\section{The model: staircase polygons}\label{sec:themodel}

\subsection{Binomial paths}\label{ssec:binomial_paths}

Let \(\L\) be a square lattice which has been rotated \(45^\circ\) from the usual orientation and scaled by \(\sqrt{2}\), so that all vertices have integer coordinates (in particular, all integer points \((\x,\y)\) with \(\x \equiv \y \,(\text{mod }2)\) are vertices). A \emph{binomial path} \(\omega\) on \(\L\) is a sequence of vertices \((\omega_0,\omega_1,\dots,\omega_n)\) such that \(\omega_i - \omega_{i-1} \in \{(1,1),(1,-1)\}\) for all \(i=1,\dots,n\). That is, \(\omega\) is a path comprising \((1,1)\) (``up'' or ``north-east'') steps and \((1,-1)\) (``down'' or ``south-east'') steps.

We will first briefly review some known results about the physics of adsorbing and pulled binomial paths. Let \(\Lp \subset \L\) be the subspace defined by \(\y \geq 0\), and define $c_n(v,h)$ to be the number of $n$-step binomial paths $\omega$ in $\Lp$ which begin at vertex $(0,0)$ and finish at vertex $(n,h)$, and which contain $v$ vertices in the line $\y=0$. Define the partition function
\[C_n(a,y) = \sum_{h,v} c_n(v,h) a^v y^h\]
and the generating function
\[C(t;a,y) = \sum_{n,v,h} c_n(v,h) t^n a^v y^h = \sum_n C_n(a,y) t^n.\]
Then~\cite[Eqn.~5.47]{vanRensburg2000Statistical}
\begin{equation}\label{eqn:binomial_paths_full_gf}
C(t;a,y) = \frac{2a\left(1-2t^2+\sqrt{1-4t^2}\right)}{\left(1-2t^2a+\sqrt{1-4t^2}\right)\left(1-2ty+\sqrt{1-4t^2}\right)}.
\end{equation}
(Note that we are associating a weight $a$ with the first vertex of a path, while this is not the case in~\cite{vanRensburg2000Statistical}.)

It is a central tenet of enumerative combinatorics that the dominant singularity (i.e.~the one closest to the origin) of a generating function determines the asymptotic behaviour of the coefficients (see e.g.~\cite{Flajolet2009Analytic}). Viewing $C(t;a,y)$ as a Taylor series in $t$ with coefficients $C_n(a,y) \in \mathbb Z[a,y]$, we can then examine the generating function for various $a,y\geq0$ in order to find the dominant singularity. Denoting this quantity by $t_c(a,y)$, we have
\begin{equation}
t_c(a,y) = \min\left\{\textstyle \frac12,\frac{y}{y^2+1},\frac{\sqrt{a-1}}{a}\right\} = \begin{cases} \frac12 &\quad a\leq 2 \text{ and } y \leq 1 \\
 \frac{y}{y^2+1} &\quad y>1 \text{ and } a<	y^2+1 \\
 \frac{\sqrt{a-1}}{a} &\quad a > 2 \text{ and } y <\sqrt{a-1}.
 \end{cases}
\end{equation}
It follows that
\begin{equation}
\psi^\textup{P}(a,y) := \lim_{n\to\infty} \frac1n\log C_n(a,y) = \max\left\{\lambda^\textup{P}(y),\kappa^\textup{P}(a)\right\} \label{eqn:psiP_maxform}
\end{equation}
where
\[\lambda^\textup{P}(y) = \begin{cases} \log2 &\quad y \leq 1 \\ \log(y^2+1)-\log y &\quad y>1\end{cases}\]
and
\[\kappa^\textup{P}(a) = \begin{cases} \log2 &\quad a \leq 2 \\ \log a-\frac12\log(a-1) & \quad a>2.\end{cases}\]

\subsection{Staircase polygons}\label{ssec:staircase_polygons}

A \emph{staircase polygon} (also known as a \emph{parallelogram polygon}) $\pi$ is a pair $(\pi^+,\pi^-)$ of two binomial paths \(\pi^+ = (\pi^+_0,\dots, \pi^+_n)\) and \(\pi^- = (\pi^-_0,\dots,\pi^-_n)\) with
\begin{itemize}
	\item \(\pi^+_0 = \pi^-_0\),
	\item \(\pi^+_n = \pi^-_n\), and 
	\item \(\pi^+_i > \pi^-_i$ for $i=1,\dots,n-1\).
\end{itemize}
That is, \(\pi\) consists of two binomial paths \(\pi^+\) and \(\pi^-\) which start and end at the same points, but otherwise with \(\pi^+\) strictly above \(\pi^-\). If the two paths each have \(n\) steps then we say that \(\pi\) has \emph{length} \(2n\).

Let \(\mathcal P^\textup{S}_{2n}\) be the set of staircase polygons of length \(2n\) which
\begin{itemize}
	\item lie entirely within \(\Lp\),
	\item contain at least one vertex in the line \(\y = 0\), and
	\item have their leftmost vertex in the line \(\x = 0\).
\end{itemize}
(The last condition is equivalent to saying that polygons in $\mathcal P^\textup{S}_{2n}$ are considered modulo translation in the $\pm \x$ directions.) Let $p^\textup{S}_{2n} = |\mathcal P^\textup{S}_{2n}|$.

For a staircase polygon $\pi\in \mathcal P^\textup{S}_{2n}$, let $v(\pi)$ be the number of vertices of $\pi$ in the line $\y = 0$. Such vertices will be called \emph{visits}. Let $h(\pi)$ be the $\y$-coordinate of $\pi^+_{\lfloor n/2 \rfloor}$. This will be called the \emph{height} of $\pi$. See \cref{fig:surface_staircase}. For $a,y\in \mathbb R^+$ we then define the partition function
\begin{equation}\label{eqn:staircase_PF}
	P^\textup{S}_{2n}(a,y) := \sum_{\pi \in \mathcal P^\textup{S}_{2n}} a^{v(\pi)}y^{h(\pi)}.
\end{equation}
(Note the difference between the weight $y$ and the ordinate $\y$ of a lattice point.)

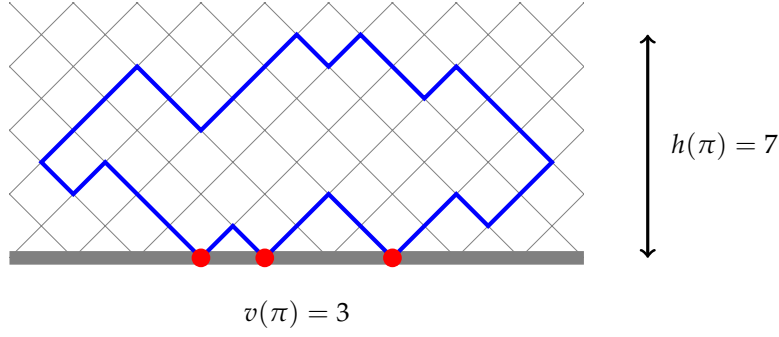
\begin{figure}
	\centering
	\begin{tikzpicture}[scale=0.6,rotate=-45]
		\begin{scope}
			\clip (6,0) -- (15,9) -- (11,13) -- (2,4) -- cycle;
			\draw[help lines] (0,0) grid (15,15);
		\end{scope}
		\draw[line width=1ex, gray] (6,0) -- (15,9);
		\draw[ultra thick, blue] (5,2) -- (6,2) -- (6,3) -- (9,3) -- (9,4) -- (10,4) -- (10,6) -- (12,6) -- (12,8) -- (13,8) -- (13,10) -- (10,10) -- (10,9) -- (8,9) -- (8,8) -- (7,8) -- (7,5) -- (6,5) -- (5,5) -- cycle;
		\node [circle, fill=red, inner sep=2.5pt] at (9,3) {};
		\node [circle, fill=red, inner sep=2.5pt] at (10,4) {};
		\node [circle, fill=red, inner sep=2.5pt] at (12,6) {};
		\draw[very thick, <->] (16,10) -- (12.5,13.5);
		\node at (14.25,11.75) [right=1ex] {\small $h(\pi)=7$};
		\node at (11,4) [below] {\small $v(\pi)=3$};
	\end{tikzpicture}
	\caption{A half-space staircase polygon of length 32, with 3 visits and height 7.}
	\label{fig:surface_staircase}
\end{figure}

We define the limiting free energy to be
\begin{equation}\label{eqn:staircase_FE}
	\psi^\textup{S}(a,y) := \lim_{n\to\infty} \frac1{2n} \log P^\textup{S}_{2n}(a,y).
\end{equation}

At this point it is not immediately obvious that this limit should even exist. In \cref{thm:psi+} we will prove its existence and give its value for all $a$ and $y$. Before stating that theorem, however, we will introduce two subsets of staircase polygons, for which certain calculations will be easier.

\subsection{Grafted staircase polygons}\label{ssec:grafted_model}

We will say a polygon $\pi\in\mathcal P^\textup{S}_{2n}$ is \emph{grafted} if
\begin{itemize}
	\item the leftmost vertex of $\pi$ has $\y$-coordinate 1, and
		\subitem $\circ$ if $n$ is even, the rightmost vertex of $\pi$ has $\y$-coordinate 1, or
		\subitem $\circ$ if $n$ is odd, the rightmost vertex of $\pi$ has $\y$-coordinate 2.
\end{itemize}
See \cref{fig:grafted_staircase}. If $n$ is even, note that by removing the two leftmost and two rightmost edges (whose locations and directions are completely fixed), a grafted polygon can be viewed as two nonintersecting Dyck paths: the lower path starting and ending on the $\x$-axis, and the upper path starting and ending on the line $\y=2$. If $n$ is odd, then doing the same almost gives two Dyck paths -- each path will end one unit higher than where it started.

\begin{figure}
	\centering
	\begin{tikzpicture}[scale=0.6, rotate=-45]
		\begin{scope}
			\clip (6,0) -- (15,9) -- (10.5,13.5) -- (1.5,4.5) -- cycle;
			\draw[help lines] (0,0) grid (15,15);
		\end{scope}
		\draw[line width=1ex, gray] (6,0) -- (15,9);
		\draw[ultra thick, blue] (6,1) -- (7,1) -- (7,2) -- (8,2) -- (8,3) -- (9,3) -- (9,5) -- (10,5) -- (10,7) -- (13,7) -- (13,8) -- (14,8) -- (14,9)  -- (11,9) -- (9,9) -- (9,8) -- (6,8) -- cycle;
		\draw[ultra thick, black] (5.8,1.2) -- (6.2,0.8);
		\draw[ultra thick, black] (5.8,0.8) -- (6.2,1.2);
		\draw[ultra thick, black] (13.8,9.2) -- (14.2,8.8);
		\draw[ultra thick, black] (13.8,8.8) -- (14.2,9.2);
		\node [circle, fill=red, inner sep=2.5pt] at (7,1) {};
		\node [circle, fill=red, inner sep=2.5pt] at (8,2) {};
		\node [circle, fill=red, inner sep=2.5pt] at (9,3) {};
		\node [circle, fill=red, inner sep=2.5pt] at (13,7) {};
		\node [circle, fill=red, inner sep=2.5pt] at (14,8) {};
		\draw[very thick, <->] (16,10) -- (12.5,13.5);
		\node at (14.25,11.75) [right=1ex] {\small $h(\pi)=7$};
		\node at (11,4) [below] {\small $v(\pi)=5$};
	\end{tikzpicture}
	\caption{A grafted staircase polygon of length 32, with 5 visits and height 7. The vertices marked with black crosses are fixed at height 1 above the surface.}
	\label{fig:grafted_staircase}
\end{figure}
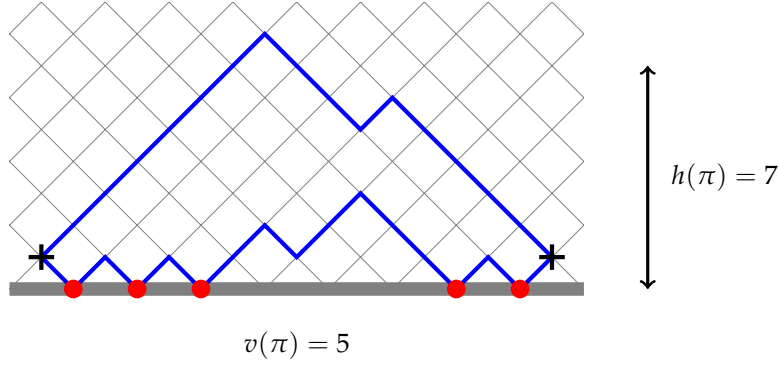

Let $\mathcal P^\textup{G}_{2n}$ be the set of grafted polygons of length $2n$. Similarly to~\eqref{eqn:staircase_PF}, we have the partition function
\begin{equation}\label{eqn:grafted_PF}
	P^\textup{G}_{2n}(a,y) := \sum_{\pi \in \mathcal P^\textup{G}_{2n}} a^{v(\pi)}y^{h(\pi)}
\end{equation}
and the free energy
\begin{equation}\label{eqn:grafted_FE}
	\psi^\textup{G}(a,y) := \lim_{n\to\infty} \frac1{2n} \log P^\textup{G}_{2n}(a,y).
\end{equation}

We will prove the following theorem.
\begin{thm}\label{thm:psiG}
	The free energy $\psi^\textup{G}(a,y)$ exists for all $a,y\in\mathbb R^+$. Moreover, it is given by
	\begin{align}
		\psi^\textup{G}(a,y) &= \textstyle \frac12 \lambda^\textup{P}\left(\sqrt{y}\right) + \frac12 \kappa^\textup{P}(a) \\
		&= \begin{cases} 
			\log 2 &\quad a \leq 2 \text{ and } y \leq 1 \\
			\frac12\log2 + \frac12\log a - \frac14\log(a-1) &\quad a > 2 \text{ and } y \leq 1 \\
			\frac12\log2 + \frac12\log(y+1) - \frac14\log y &\quad a \leq 2 \text{ and } y > 1 \\
			\frac12\log a - \frac14\log(a-1) + \frac12\log(y+1) - \frac14\log y &\quad a > 2 \text{ and } y > 1.
		\end{cases} \label{eqn:psiG}
	\end{align}
\end{thm}
We will denote the four regions in the $(a,y)$-quarter-plane as $G_\textup{I}, G_\textup{II}, G_\textup{III}$ and $G_\textup{IV}$ respectively. These correspond to free ($\mathbf V = \mathbf H = 0$), adsorbed ($\mathbf V>0$), ballistic ($\mathbf H>0$) and ``mixed'' ($\mathbf V>0$, $\mathbf H>0$) phases respectively.

\subsection{Centred staircase polygons}\label{ssec:centred_model}

We will say that a polygon $\pi\in\mathcal P^\textup{S}_{2n}$ is \emph{centred} if the vertex $\pi^-_{\lfloor n/2 \rfloor}$ has $\y$-coordinate 0. See \cref{fig:centred_staircase}. Let $\mathcal P^\textup{C}_{2n}$ be the set of centred polygons of length $2n$. We again have a partition function
\begin{equation}\label{eqn:centred_PF}
	P^\textup{C}_{2n}(a,y) := \sum_{\pi \in \mathcal P^\textup{C}_{2n}} a^{v(\pi)}a^{h(\pi)}
\end{equation}
and free energy
\begin{equation}\label{eqn:centred_FE}
	\psi^\textup{C}(a,y) := \lim_{n\to\infty} \frac1{2n} \log P^\textup{C}_{2n}(a,y).
\end{equation}
\begin{thm}\label{thm:psiC}
	The free energy $\psi^\textup{C}(a,y)$ exists for all $a,y\in\mathbb R^+$. Moreover, it is given by
	\begin{align}
		\psi^\textup{C}(a,y) &= \textstyle \frac12 \lambda^\textup{P}\left(\sqrt{y}\right) + \frac12 \max\left\{\lambda^\textup{P}\left(\sqrt{y}\right),\kappa^\textup{P}(a)\right\} \\
		&=\begin{cases} 
			\log 2 &\quad a \leq 2 \text{ and } y \leq 1 \\
			\frac12\log2 + \frac12\log a - \frac14\log(a-1) &\quad a > 2 \text{ and } y \leq 1 \\
			\log(y+1) - \frac12\log y &\quad y > 1 \text{ and } a \leq y + 1 \\
			\frac12\log a - \frac14\log(a-1) + \frac12\log(y+1) - \frac14\log y &\quad y > 1 \text{ and } a > y + 1.
		\end{cases}\label{eqn:psiC}
	\end{align}
\end{thm}
We will denote these four regions in the $(a,y)$-quarter-plane as $C_\textup{I}, C_\textup{II}, C_\textup{III}$ and $C_\textup{IV}$ respectively. Note that $C_\textup{I} = G_\textup{I}$ and $C_\textup{II} = G_\textup{II}$.

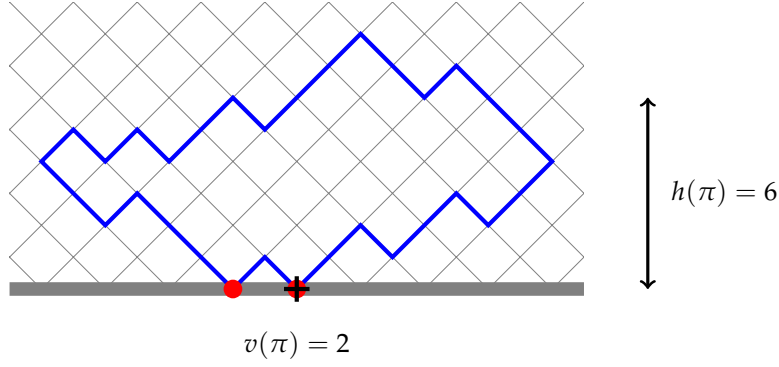
\begin{figure}
	\centering
	\begin{tikzpicture}[scale=0.6, rotate=-45]
		\begin{scope}
			\clip (6.5,0.5) -- (15.5,9.5) -- (11,14) -- (2,5) -- cycle;
			\draw[help lines] (0,0) grid (16,16);
		\end{scope}
		\draw[line width=1ex, gray] (6.5,0.5) -- (15.5,9.5);
		\draw[ultra thick, blue] (5,3) -- (7,3) -- (7,4) -- (10,4) -- (10,5) -- (11,5) -- (11,7) -- (12,7) -- (12,9) -- (13,9) -- (13,11) -- (10,11) -- (10,10) -- (8,10) -- (8,7) -- (7,7) -- (7,5) -- (6,5) -- (6,4) -- (5,4) -- cycle;
		\node [circle, fill=red, inner sep=2.5pt] at (10,4) {};
		\node [circle, fill=red, inner sep=2.5pt] at (11,5) {};
		\draw[ultra thick, black] (10.8,5.2) -- (11.2,4.8);
		\draw[ultra thick, black] (10.8,4.8) -- (11.2,5.2);
		\draw[very thick, <->] (16.5,10.5) -- (13.5,13.5);
		\node at (15,12) [right=1ex] {\small $h(\pi)=6$};
		\node at (11.5,4.5) [below]{\small $v(\pi)=2$};
	\end{tikzpicture}
	\caption{A centred staircase polygon of length 32, with 2 visits and height 6. The vertex marked with a black cross is fixed in the surface.}
	\label{fig:centred_staircase}
\end{figure}

The main theorem of this paper is the following, which essentially says that the restriction from staircase polygons to centred polygons is weak enough so as to have no effect on the free energy.
\begin{thm}\label{thm:psi+}
	The free energy $\psi^\textup{S}(a,y)$ exists for all $a,y\in\mathbb R^+$. Moreover, 
	\begin{equation}\label{eqn:psi+psiC}
		\psi^\textup{S}(a,y) = \psi^\textup{C}(a,y).
	\end{equation}
\end{thm}

We note here that the problem of adsorbing staircase polygons has been studied in the past by Janse van Rensburg~\cite{vanRensburg1999Adsorbing} (see also \cite[Section~5.6]{vanRensburg2000Statistical}). One result obtained there is that the free energy for a certain subclass of staircase polygons (those with leftmost vertex at ordinate $\y=1$; one might call them ``semi-grafted''), the free energy $\kappa^\textup{SG}(a)$ is given by
\begin{align}
\kappa^\textup{SG}(a) &= \textstyle\frac12\log2 + \frac12\kappa^\textup{P}(a) \\
&= \begin{cases} \log 2 &\quad a \leq 2 \\ \frac12 \log2 + \frac12 \log a-\frac14\log(a-1) &\quad a>2.\end{cases}
\end{align}
Observe then that for $y\leq 1$,
\[\psi^\textup{G}(a,y) = \psi^\textup{C}(a,y) = \psi^\textup{S}(a,y) = \kappa^\textup{SG}(a).\]

In the next section we will employ some of the ideas used in~\cite{vanRensburg1999Adsorbing,vanRensburg2000Statistical} in order to generalise to the adsorbing and pulled case.

\section{Pairs of paths}\label{sec:pairs-of-paths}

A key tool in computing the partition functions and free energies defined in the previous section will be a formulation of staircase polygons -- with given length, height and number of visits -- in terms of pairs of nonintersecting lattice paths. Such a technique is well known -- see for example~\cite[Section~5.6]{vanRensburg2000Statistical}. Indeed, we will take as our starting point a result from~\cite{vanRensburg2000Statistical}.

Let $P$ be a configuration of two nonintersecting binomial paths of length $n$ in $\Lp$, with the lower path starting at $(0,j_0-l_0)$ and finishing at $(n,j-l)$, and the upper path starting at $(0,j_0+l_0+2)$ and finishing at $(n,j+l+2)$. Furthermore, suppose the lower path accumulates a weight $a$ for every vertex in the boundary $\y=0$. See \cref{fig:two_paths}. Let $r_n(j_0,l_0;j,l)$ be the total weight of all such configurations.

\begin{figure}
	\centering
	\begin{tikzpicture}[scale=0.6, rotate=-45]
		\begin{scope}
			\clip (6.5,0.5) -- (15.5,9.5) -- (11,14) -- (2,5) -- cycle;
			\draw[help lines] (0,0) grid (16,16);
		\end{scope}
		\draw[line width=1ex, gray] (6.5,0.5) -- (15.5,9.5);
		\draw[ultra thick, blue] (7,1) -- (7,2) -- (8,2) -- (8,3) -- (9,3) -- (9,5) -- (10,5) -- (10,7) -- (13,7) -- (13,8) -- (14,8) -- (14,10);
		\draw[ultra thick, blue] (5,3) -- (6,3) -- (6,4) -- (7,4) -- (7,6) -- (8,6) -- (8,7) -- (9,7) -- (9,10) -- (10,10) -- (10,12) -- (11,12) -- (11,13);		
		\node [circle, fill=red, inner sep=2.5pt] at (7,1) {};
		\node [circle, fill=red, inner sep=2.5pt] at (8,2) {};
		\node [circle, fill=red, inner sep=2.5pt] at (9,3) {};
		\node [circle, fill=red, inner sep=2.5pt] at (13,7) {};
		\node [circle, fill=red, inner sep=2.5pt] at (14,8) {};
		\node at (11.5,13.5) [right] {\small $j+l+2=8$};
		\node at (14.5,10.5) [right] {\small $j-l = 2$};
		\node at (6.5,0.5) [left] {\small $j_0-l_0 = 0$};
		\node at (4.5,2.5) [left] {\small $j_0+l_0+2 = 4$};
	\end{tikzpicture}
	\caption{Two paths of length 16, starting at $(0,0)$ and $(0,4)$ and finishing at $(16,2)$ and $(16,8)$, with the bottom path accumulating visits to the surface. This pair contributes a weight of $a^5$ to $r_{16}(1,1;4,2)$.}
	\label{fig:two_paths}
\end{figure}
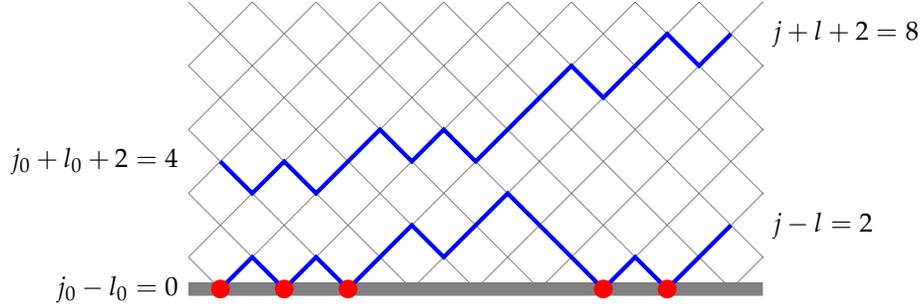

We denote by $\mathrm{CT}_\mathcal{S}[X]$ the constant term of the Laurent series $X$, with respect to the set of variables $\mathcal S$. Also let $\overline x = 1/x$.

\begin{lem}[{\cite[Eqn.~(5.168)]{vanRensburg2000Statistical}}] The partition function $r_n(j_0,l_0;j,l)$ is given by
	\begin{multline}\label{eqn:rn_CTform}
		r_n(j_0,l_0;j,l) = \mathrm{CT}_{\{\z_1,\z_2\}} \left[\Lambda^n \z_1^j \z_2^l \left(\oz_1^{j_0}\oz_2^{l_0} + \mathfrak{T}_1\oz_1^{l_0}\oz_2^{j_0} - \oz_1^{j_0}\z_2^{l_0+2} - \mathfrak{T}_2\oz_1^{l_0}\z_2^{j_0+2} \right.\right. \\ \left.\left. 
		- \mathfrak{T}_1\mathfrak{T}_2\z_1^{j_0+2}\oz_2^{l_0}  - \mathfrak{T}_1\z_1^{l_0+2}\oz_2^{j_0} + \mathfrak{T}_1\mathfrak{T}_2\z_1^{j_0+2}\z_2^{l_0+2} + \mathfrak{T}_2\z_1^{l_0+2}\z_2^{j_0+2}\right)\right],
	\end{multline}
	where
	\[\Lambda = \z_1 + \oz_1 + \z_2 + \oz_2,\]
	\[\mathfrak{T}_1 = - \frac{\Lambda - a\left(\oz_1 + \z_2\right)}{\Lambda - a\left(\z_1 + \oz_2\right)} \qquad \text{and}\qquad \mathfrak{T}_2 = - \frac{\Lambda - a\left(\oz_1 + \oz_2\right)}{\Lambda - a\left(\z_1 + \z_2\right)}.\]
\end{lem}
This result is obtained by finding a recurrence relation satisfied by $r_n$. This is straightforward: one can construct pairs of paths of length $n$ by appending steps to pairs of paths of length $n-1$, taking care to make sure that the two paths do not intersect and that the bottom path does not step below the surface, as well as accounting for the weight $a$ when the bottom path touches the surface. This recurrence can then be solved with a general Bethe ansatz, where Cauchy's integral formula allows one to write a contour integral as the constant term of a Laurent series.

The most obvious way to obtain a staircase polygon from a pair of binomial paths is to take $l_0 = l = 0$, so that the two paths start and end as near to each other as possible. The addition of two edges on the left and two edges on the right then joins the paths up to form a polygon. This would correctly account for the $a$ weight in the partition function; however, it gives us no information about the height of the resulting polygon. 

We will thus use a different strategy: taking \emph{two} pairs of paths of length $n$, both with $l_0=0$ and with matching values of $j$ and $l$, then reflecting one of the pairs through the vertical axis, gluing the two pairs together, and finally adding two edges at the left and right to close up the whole thing into a polygon. The number of visits for the resulting polygon is the sum of the numbers of visits of the two pairs (unless $j=l$, in which case one visit has been counted twice), while the height of the resulting polygon is $j+l+2$.

We thus set $l_0=0$ in~\eqref{eqn:rn_CTform}. Following~\cite{vanRensburg2000Statistical}, we perform the change of variables $\z_1 = \e_1\e_2$ and $\z_2 = \e_1\oe_2$. Then $\Lambda = (\e_1+\oe_1)(\e_2+\oe_2)$, and if we define
\[\mathcal H(\eta) = - \frac{\eta^2 - (a-1)}{1 - (a-1)\eta^2}\]
then $\mathfrak{T}_1 = \mathcal H(\e_2)$ and $\mathfrak{T}_2 = \mathcal H(\e_1)$. (Note that this corrects~\cite[Eqn.~(5.170)]{vanRensburg2000Statistical}.) We will write $\mathcal{H}_1 = \mathcal H(\e_1)$ and $\mathcal{H}_2 = \mathcal H(\e_2)$ for short. Then
\begin{multline}\label{eqn:rn_CTform_changed}
	r_n(j_0,0;j,l) = \mathrm{CT}_{\{\e_1,\e_2\}}\left[\Lambda^n \e_1^{j+l}\e_2^{j-l} \left(\oe_1^{j_0}\oe_2^{j_0} + \mathcal{H}_2\oe_1^{j_0}\e_2^{j_0} - \oe_1^{j_0-2}\oe_2^{j_0+2} - \mathcal{H}_1\e_1^{j_0+2}\oe_2^{j_0+2} \right.\right. \\ \left. \left. 
	-\mathcal{H}_1\mathcal{H}_2\e_1^{j_0+2}\e_2^{j_0+2} - \mathcal{H}_2\oe_1^{j_0-2}\e_2^{j_0+2} + \mathcal{H}_1\mathcal{H}_2 \e_1^{j_0+4}\e_2^{j_0} + \mathcal{H}_1\e_1^{j_0+4}\oe_2^{j_0} \right)\right].
\end{multline}

\subsection{Grafted polygons via pairs of paths}\label{ssec:grafted_paths}

Let $s_n(k,m) = r_n\left(0,0;\frac{k+m}{2},\frac{k-m}{2}\right)$, with the specialisation $s^1_n(k,m) = s_n(k,m)|_{a=1}$. This can be extracted directly from~\eqref{eqn:rn_CTform_changed}; it is~\cite[Eqn.~(5.172)]{vanRensburg2000Statistical}
\begin{align}
	s^1_n(k,m) &= \mathrm{CT}_{\{\e_1,\e_2\}}\left[(\e_1+\oe_1)^n(\e_2+\oe_2)^n \e_1^{k}\e_2^{m}(1-\e_1^2)(1-\e_2^2)(\e_1^2-\e_2^2)(\e_1^2-\oe_2^2)\right] \notag \\
	&= \frac{(k+3)(m+1)(k-m+2)(k+m+4)}{4(n+1)(n+2)(n+3)^2} \binom{n+3}{\frac12(n+k+6)}\binom{n+3}{\frac12(n+m+4)}. \label{eqn:snl_a1_exact}
\end{align}
The general $a$ case is then obtained by expanding $\mathcal{H}_1$ and $\mathcal{H}_2$ in~\eqref{eqn:rn_CTform_changed} as power series; we have~\cite[Eqn.~(5.173)]{vanRensburg2000Statistical}
\begin{equation}\label{eqn:snl_exact}
	s_n(k,m) = a \sum_{w=0}^n (a-1)^w \sum_{p=0}^w s_n^1(k+2p,m+2w-2p).
\end{equation}
Note that $s_n(k,m)$ should only be nonzero when $0 \leq m \leq k \leq n$ and $k \equiv m \equiv n \,(\text{mod }2)$; in other words, the binomial coefficients should be taken to be zero if either argument is not a natural number.

The arguments outlined above give the following lemma.
\begin{lem}\label{lem:even_grafted_from_paths}
	For even $n$, the partition function $P^\textup{G}_{2n}(a,y)$ is given by
	\begin{equation}\label{eqn:even_grafted_from_paths}
		P^\textup{G}_{2n}(a,y) = \frac1a \sum_{k=0}^{\frac{n-2}{2}} \left(s_{\frac{n-2}{2}}(k,0)\right)^2y^{k+2} + \sum_{k=0}^{\frac{n-2}{2}}\sum_{m=1}^k \left(s_{\frac{n-2}{2}}(k,m)\right)^2y^{k+2}.
	\end{equation}
\end{lem}

\begin{proof}
	$s_{\frac{n-2}{2}}(k,m)$ is the total weight of pairs of paths of length $\frac{n-2}{2}$, starting at $(0,0)$ and $(0,2)$ and ending at $\left(\frac{n-2}{2},m\right)$ and $\left(\frac{n-2}{2},k+2\right)$. For given $k$ and $m$, we can take two such pairs, reflect one of them through the vertical line $\x=\frac{n-2}{2}$, and take the whole thing to be one large pair of paths, the lower going from $(0,0)$ to $(n-2,0)$ and the upper going from $(0,2)$ to $(n-2,2)$. The surface weight of this new pair of paths (that is, $a$ raised the power of the total number of visits) will be equal to the product of the weights of the original two pairs, unless $m=0$, in which case it will be off by a factor of $a$. The height of the upper ``middle'' vertex will be $k+2$. Since each of the original paths had length $\frac{n-2}{2}$, the total length of this new pair is $2n-4$; adding two edges at the left and right to close things up gives total length of $2n$. 
\end{proof}

To deal with the odd $n$ case, define $\hat s_n(k,m) = r_n\left(1,0;\frac{k+m}{2},\frac{k-m}{2}\right)$, with the specialisation $\hat s_n^1(k,m) = \hat s_n(k,m)|_{a=1}$. This time we have
\begin{align}
	\hat s^1_n(k,m) &= \mathrm{CT}_{\{\e_1,\e_2\}}\left[(\e_1+\oe_1)^n(\e_2+\oe_2)^n \e_1^{k-1}\e_2^{m-1}(1-\e_1^4)(1-\e_2^4)(\e_1^2-\e_2^2)(\e_1^2-\oe_2^2)\right] \notag \\
	&= \frac{(k+3)(m+1)(k-m+2)(k+m+4)}{4(n+2)(n+3)(n+4)^2}\binom{n+4}{\frac12(n+k-1)}\binom{n+4}{\frac12(n+m+5)}.\label{eqn:hatsnl_a1_exact}
\end{align}
Then $\hat s_n(k,m)$ expands in the same way as~\eqref{eqn:snl_exact}:
\begin{equation}\label{eqn:hatsnl_exact}
	\hat s_n(k,m) = \sum_{w=0}^n (a-1)^w \sum_{p=0}^w \hat s_n^1(k+2p,m+2w-2p).
\end{equation}
This time $\hat s_n(k,m)\neq 0$ only if $0\leq m \leq k \leq n+1$ and $k \equiv m \equiv n+1 \,(\text{mod }2)$.

The following can be proved in the same way as \cref{lem:even_grafted_from_paths}.
\begin{lem}\label{lem:odd_grafted_from_paths}
	For odd $n$, the partition function $P^\textup{G}_{2n}(a,y)$ is given by
	\begin{equation}\label{eqn:odd_grafted_from_paths}
		P^\textup{G}_{2n}(a,y) = \frac1a \sum_{k=0}^{\frac{n-3}{2}} s_{\frac{n-3}{2}}(k,0)\hat s_{\frac{n-1}{2}}(k,0) y^{k+2} + \sum_{k=0}^{\frac{n-3}{2}}\sum_{m=1}^k s_{\frac{n-3}{2}}(k,m)\hat s_{\frac{n-1}{2}}(k,m) y^{k+2}.
	\end{equation}
\end{lem}

\subsection{Centred polygons via pairs of paths}\label{ssec:centred_paths}

The methods of \cref{ssec:grafted_paths} can be repeated to give an explicit formula for the partition function $P^\textup{C}_{2n}(a,y)$. It requires taking $j=l$ in~\eqref{eqn:rn_CTform_changed} and then computing the constant term as a function of $n,j_0$ and $l$. However, the significant cancellations which led to the nice formulations of~\eqref{eqn:snl_a1_exact} and~\eqref{eqn:hatsnl_a1_exact} cannot be applied here, and as a result the final formula is much more complicated.

This would not be so bad if it were still possible to work with said formula (or, more precisely, for \textsc{Mathematica} to do so). However, we have not been able to successfully carry out all the calculations outlined in \cref{sec:grafted_asymps} for centred polygons. For this reason we will not go into any details regarding the exact formula for $P^\textup{C}_{2n}(a,y)$.

\section{Exact asymptotics for grafted polygons}\label{sec:grafted_asymps}

In this section we will compute the asymptotic behaviour of $P^\textup{G}_{2n}(a,y)$ as $n\to\infty$ for all $a$ and $y$, in the process providing a proof of~\cref{thm:psiG}. There are nine different $(a,y)$ regimes which each give different asymptotics. For brevity we will not go through all nine, and will instead provide details for three different cases only. We will also focus only on even $n$, as the procedure for odd $n$ is virtually identical.

\subsection{The multicritical point: $(a,y) = (2,1)$}\label{ssec:a2_y1}

The partition function $P^\textup{G}_{2n}(a,y)$ is (as per~\eqref{lem:even_grafted_from_paths}) given by a triple sum (over $k,w,p$) plus a quadruple sum (over $k,m,w,p$), which we write as $P^\textup{G}_{2n}(a,y) = \frac1a P^\textup{G1}_{2n}(a,y) + P^\textup{G2}_{2n}(a,y)$. We first focus on $P^\textup{G1}_{2n}(a,y)$. Note that this is nonzero only if $n\equiv 2\,(\text{mod }4)$, so we assume this to be the case.

We begin by observing that the dominant contributions to $P^\textup{G1}_{2n}(2,1)$ come from values of $k,w$ and $p$ which are all $O\left(\sqrt{n}\right)$. (This is easily seen with numerical data, or by analytically maximising the summands.) We then take $s^1_n(k+2p,2w-2p)$ and apply Stirling's approximation
\begin{equation}\label{eqn:stirling}
	n! \sim \sqrt{2\pi n} \left(\frac{n}{e}\right)^n.
\end{equation}

As $n\to\infty$ the asymptotics of the sum can be found by approximating it with an integral. We thus set $n=\tilde{n}^2$, and then $k=\kappa \tilde n, w= \omega \tilde n$ and $p = \rho \tilde n$. The resulting expression is then expanded about $\tilde n=\infty$ (equivalently, set $\tilde n=1/\epsilon$ and expand about $\epsilon =0$). We then integrate the dominant term first in $\rho$ over $(0,\omega)$ and then in $\omega$ over $(0,\infty)$, giving
\begin{equation}\label{eqn:int_a2_y1}
	s_{\tilde n^2}(\kappa\tilde n,0)|_{a=2} \sim \frac{32\kappa^2 e^{-\kappa^2/2}}{\pi} \times \frac{4^{\tilde n^2}}{\tilde n^4}.
\end{equation}
Squaring the above, integrating in $\kappa$ over $(0,\infty)$ and returning to $n$ gives
\[\frac{96}{\pi^{3/2}} \times \frac{16^n}{n^{7/2}}.\]
However, this is not correct, because $s_n(k,m)$ is only nonzero if $k \equiv n\,(\text{mod }2)$, but in integrating~\eqref{eqn:int_a2_y1} in $\kappa$ we did not take this into account. We must thus divide this result by 2 to obtain correct asymptotics. Finally adjusting $n \mapsto \frac{n-2}{2}$, we find
\begin{equation}\label{eqn:G1_a2_y1_asymps}
	P^\textup{G1}_{2n}(2,1) \sim \frac{24\sqrt{2}}{\pi^{3/2}} \times \frac{4^n}{n^{7/2}}, \qquad\qquad n\equiv2\,(\text{mod }4).
\end{equation}

For $P^\textup{G2}_{2n}(2,1)$, the dominant contribution to the summands is when all of $k,m,w$ and $p$ are $O\left(\sqrt{n}\right)$. We follow the same procedure as above, with the additional substitution $m = \mu \tilde n$ and integral in $\mu$ over $(0,\kappa)$. Instead of being off by a factor of 2 at the end, we are off by a factor of 4, as we only want terms with $k \equiv m \equiv n\,(\text{mod }2)$. We arrive at
\begin{equation}\label{eqn:G2_a2_y1_asymps}
	P^\textup{G2}_{2n}(2,1) \sim \frac{4}{\pi} \times \frac{4^n}{n^3}, \qquad\qquad n\text{ even}.
\end{equation}

Note that $P^\textup{G2}_{2n}(2,1)$ is of strictly greater order than $P^\textup{G1}_{2n}(2,1)$. 

\subsection{Weakly attractive or repulsive surface and no force: $a<2$ and $y=1$}\label{ssec:asmall_y1}

When $a<2$ the energetic gain from a surface visit is less than the entropic loss. As a result, the dominant contribution to $P^\textup{G1}_{2n}(a,1)$ is now from terms with $w$ and $p$ in $O(1)$, with $k$ still $O\left(\sqrt{n}\right)$. We start with $s^1_n(k,0)$ as before, and again set $n=\tilde n^2$ and $k=\kappa \tilde n$, but leave $w$ and $p$ as constants (with respect to $\tilde n$) when expanding about $\tilde n=\infty$. The leading term is then summed over $p=0,1,\dots,w$ and $w=0,1,\dots$. We then square this and integrate in $\kappa$ over $(0,\infty)$. Finally, we convert back to $n$, divide by 2 (for the same reason as above), and adjust $n\mapsto \frac{n-2}{2}$. This gives
\begin{equation}\label{eqn:G1_asmall_y1_asymps}
	P^\textup{G1}_{2n}(a,1) \sim \frac{1920\sqrt{2}a^4}{\pi^{3/2}(a-2)^6} \times \frac{4^n}{n^{13/2}}, \qquad\qquad n\equiv2\,(\text{mod }4), \,a<2.
\end{equation}

The calculation for $P^\textup{G2}_{2n}(a,1)$ is analogous:
\begin{equation}\label{eqn:G2_asmall_y1_asymps}
	P^\textup{G2}_{2n}(a,1) \sim \frac{48a^2}{\pi(a-2)^4} \times \frac{4^n}{n^5}, \qquad\qquad n\text{ even}, \,a<2.
\end{equation}

\subsection{Strongly attractive surface with pulling force: $a>2$ and $y>1$}\label{ssec:abig_ybig}

In this regime, polygons with a large number of visits and a large height will dominate the partition function. For grafted polygons, there is no conflict here, because the pulling force is only ``felt'' by the upper walk and the visits are only ``felt'' by the lower walk.

For $P^\textup{G1}_{2n}(a,y)$, we first think of $\left(s_n(k,0)\right)^2y^{k+2}$ as $\left(s_n(k,0)y^{\frac{k+2}{2}}\right)^2$, and thus consider a pair of paths with a weight $a$ associated with surface visits and a weight $\sqrt{y}$ associated with the height of the upper path's endpoint. We can check (numerically or otherwise) that the dominant contribution comes from summands with $k$ and $w$ of $O(n)$ and $p$ of $O(1)$. If we set $k=\gamma n$ and $w = \delta n$, then we need to find the values of $\gamma$ and $\delta$ which maximise the asymptotic expression for the summands.

To do this, we take $s^1_n(k+2p,2w-2p)(a-1)^w y^{\frac{k+2}{2}}$, apply Stirling's approximation, make the above substitutions for $k$ and $w$, expand about $n=\infty$ and keep the leading term. Then take the logarithm, divide by $n$, separately take a derivative with respect to $\gamma$ and with respect to $\delta$, and finally send $n\to\infty$ in each. This results in
\begin{equation}\label{eqn:finding_gamma_delta}
	\frac12\log\left(\frac{y(1-\gamma)}{1+\gamma}\right) \qquad \text{and}\qquad \log\left(\frac{(a-1)(1-2\delta)}{1+2\delta}\right).
\end{equation}
The values of $\gamma$ and $\delta$ which set these two quantities to 0, and thus maximise the asymptotic growth rate, are $\frac{y-1}{y+1}$ and $\frac{a-2}{2a}$ respectively. Furthermore, it can be verified that the peak in the asymptotic growth rate, as a function of $k$ or of $w$, is of width $O\left(\sqrt{n}\right)$.

We then repeat the procedure of \cref{ssec:a2_y1,ssec:asmall_y1}, this time with $n=\tilde n^2, k = \gamma \tilde n^2 + \kappa \tilde n$ and $w=\delta\tilde n^2 + \omega \tilde n$, and integrating in both $\kappa$ and $\omega$ over $(-\infty,\infty)$. This results in
\begin{equation}\label{eqn:G1_abig_ybig_asymps}
	\begin{split} P^\textup{G1}_{2n}(a,y) \sim \frac{(a-2)^2(y-1)^2((a-1)y-1)^2}{\sqrt{2\pi}(a-1)^3y^{3/2}(y+1)} &\times \frac{1}{\sqrt{n}}\left(\frac{a(y+1)}{\sqrt{(a-1)y}}\right)^n, \\ &\qquad\qquad n \equiv 2\,(\text{mod }4), \, a>2, \, y>1. \end{split}
\end{equation}

For $P^\textup{G2}_{2n}(a,y)$, we must be careful. The dominant contribution is with $k,w$ and $p$ as above and with $m=O(1)$ (the strongly attractive surface means the bottom path tends to stay very close). But since $m$ must be of the same parity as $n$ for $s_n(k,m)$ to be nonzero, the even and odd cases have slightly different asymptotics. Eventually, one finds
\begin{equation}\label{eqn:G2_abig_ybig_asymps_even}
	\begin{split} P^\textup{G2}_{2n}(a,y) \sim \frac{(a-2)(y-1)^2((a-1)y-1)^2}{\sqrt{2\pi}a(a-1)^3y^{3/2}(y+1)} &\times \frac{1}{\sqrt{n}}\left(\frac{a(y+1)}{\sqrt{(a-1)y}}\right)^n, \\ &\qquad\qquad n \equiv 2\,(\text{mod }4), \, a>2, \, y>1 \end{split}
\end{equation}
and
\begin{equation}\label{eqn:G2_abig_ybig_asymps_odd}
	\begin{split} P^\textup{G2}_{2n}(a,y) \sim \frac{(a-2)(y-1)^2((a-1)y-1)^2}{\sqrt{2\pi}a(a-1)^2y^{3/2}(y+1)} &\times \frac{1}{\sqrt{n}}\left(\frac{a(y+1)}{\sqrt{(a-1)y}}\right)^n, \\ &\qquad\qquad n \equiv 0\,(\text{mod }4), \, a>2, \, y>1. \end{split}
\end{equation}

\subsection{Summary for all regions}\label{ssec:grafted_summary}

The above calculations can be repeated for odd $n$ using~\eqref{eqn:odd_grafted_from_paths}; the growth rates and critical exponents do not change, but in general the amplitudes are different. For brevity we omit details.

For the remaining regions in the $(a,y)$ plane, we summarise the results in \cref{table:PG_asymps}, giving only the exponential growth and the critical exponent for $P^\textup{G}_{2n}(a,y)$. (Note that for all regions with $a>2$ and/or $y<1$, the amplitude depends on the value of $n\,(\text{mod }4)$ (as in \cref{ssec:abig_ybig}), but the growth rate and exponent are independent of this.)
\begin{table}[h]
\centering
\begin{tabular}{| C{1.3cm} | C{4.1cm} | C{4.1cm} | C{4.1cm} |}
\cline{2-4}
\multicolumn{1}{C{1.3cm} |}{} & \btm a<2 \etm & \btm a=2 \etm & \btm a>2 \etm \\ \hline
\btm y>1 \etm & \btm n^{-2}\times\left(\frac{2(y+1)}{\sqrt{y}}\right)^n \etm & \btm n^{-1}\times \left(\frac{2(y+1)}{\sqrt{y}}\right)^n \etm & \btm n^{-1/2}\times\left(\frac{a(y+1)}{\sqrt{(a-1)y}}\right)^n \etm \\ \hline
\btm y=1 \etm & \btm n^{-5}\times 4^n \etm & \btm n^{-3}\times 4^n \etm & \btm n^{-3/2}\times\left(\frac{2a}{\sqrt{a-1}}\right)^n \etm \\ \hline
\btm y<1 \etm & \btm n^{-10}\times 4^n \etm & \btm n^{-6}\times4^n \etm & \btm n^{-3}\times\left(\frac{2a}{\sqrt{a-1}}\right)^n \etm \\ \hline
\end{tabular}
\caption{The dominant asymptotic behaviour of $P^\textup{G}_{2n}(a,y)$ for all values of $a,y>0$, with amplitudes omitted.}
\label{table:PG_asymps}
\end{table}

By taking logs, dividing by $2n$ and taking $n\to\infty$, we exactly obtain \cref{thm:psiG}.

The regions $G_\textup{I}$, $G_\textup{II}$, $G_\textup{III}$ and $G_\textup{IV}$ correspond to the four different phases for grafted polygons: free, adsorbed, ballistic, and mixed (adsorbed and ballistic) respectively. See \cref{fig:phase_diagrams}.

\section{Lower bounds for centred polygons}\label{sec:centred_lbs}

In \cref{sec:grafted_asymps} we found the exact dominant asymptotics for $P^\textup{G}_{2n}(a,y)$, thus verifying \cref{thm:psiG} directly. For centred polygons, things are not so straightforward. The partition function $P^\textup{C}_{2n}(a,y)$ can be computed exactly as per \cref{sec:pairs-of-paths}, but the final form is sufficiently complicated so that we have not been able to obtain exact asymptotics. Instead, we will use several techniques to find upper and lower bounds for $P^\textup{C}_{2n}(a,y)$, and whose asymptotics will be found to be sharp. In this section we focus on lower bounds.

\subsection{Grafted and centred polygons}\label{ssec:G_and_C}

A slight modification to the calculations of \cref{sec:grafted_asymps} will give a sharp lower bound in regions $G_\textup{I}, G_\textup{II}$ and $G_\textup{IV}$.

We will say that a polygon $\pi\in \mathcal P^\textup{S}_{2n}$ is \emph{grafted and centred} if 
\begin{itemize}
	\item the vertex $\pi^-_{\lfloor n/2 \rfloor}$ has $\y$-coordinate 0, and
		\subitem $\circ$ if $n \equiv 0\,(\text{mod }4)$, the leftmost and rightmost vertices of $\pi$ have $\y$-coordinate 2,
		\subitem $\circ$ if $n \equiv 1\,(\text{mod }4)$, the leftmost (resp.~rightmost) vertex of $\pi$ has $\y$-coordinate 2 (resp.~1),
		\subitem $\circ$ if $n \equiv 2\,(\text{mod }4)$, the leftmost and rightmost vertices of $\pi$ have $\y$-coordinate 1, or
		\subitem $\circ$ if $n \equiv 3\,(\text{mod }4)$, the leftmost (resp.~rightmost) vertex of $\pi$ has $\y$-coordinate 1 (resp.~2).
\end{itemize}
That is, $\pi$ is grafted and centred if it is centred and if its leftmost and rightmost vertices have the smallest possible $\y$-coordinates. Let $\mathcal P^\textup{GC}_{2n}$ be the set of grafted and centred polygons of length $2n$, with associated partition function $P^\textup{GC}_{2n}(a,y)$.

\begin{lem}\label{lem:GC_comparison}
The free energy
\begin{equation}\label{eqn:GC_fe}
	\psi^\textup{GC}(a,y) = \lim_{n\to\infty} \frac{1}{2n}\log P^\textup{GC}_{2n}(a,y)
\end{equation}
exists and is equal to $\psi^\textup{G}(a,y)$. Furthermore,
\begin{equation}\label{eqn:GC_vs_C}
	\psi^\textup{G}(a,y) = \psi^\textup{GC}(a,y) \leq \liminf_{n\to\infty} \frac{1}{2n}\log P^\textup{C}_{2n}(a,y).
\end{equation}
\end{lem}

\begin{proof}[Proof (sketch)]
Since grafted and centred polygons are a subset of centred polygons, we have $P^\textup{GC}_{2n}(a,y) \leq P^\textup{C}_{2n}(a,y)$ for all $a,y$. The inequality in~\eqref{eqn:GC_vs_C} immediately follows.

For the first part of the proof, observe that $P^\textup{GC}_{2n}(a,y)$ can be computed exactly using the quantities defined in \cref{sec:pairs-of-paths}. Namely,
\begin{equation}\label{eqn:PGC_exact}
	P^\textup{GC}_{2n}(a,y) = \begin{cases}
	\frac1a \sum_{k=0}^{\frac{n}{2}}\left(\hat s_{\frac{n-2}{2}}(k,0)\right)^2 y^{k+2} & \quad n \equiv 0\,(\text{mod }4) \\
	\frac1a \sum_{k=0}^{\frac{n-1}{2}}\hat s_{\frac{n-3}{2}}(k,0)s_{\frac{n-1}{2}}(k,0) y^{k+2} & \quad n \equiv 1\,(\text{mod }4) \\
	\frac1a \sum_{k=0}^{\frac{n-2}{2}}\left(s_{\frac{n-2}{2}}(k,0)\right)^2 y^{k+2} & \quad n \equiv 2\,(\text{mod }4) \\
	\frac1a \sum_{k=0}^{\frac{n-3}{2}}s_{\frac{n-3}{2}}(k,0)\hat s_{\frac{n-1}{2}}(k,0) y^{k+2} & \quad n \equiv 3\,(\text{mod }4).
	\end{cases}
\end{equation}
The asymptotics can then be computed in the same way as in \cref{sec:grafted_asymps} (in fact, things are simpler, since there is no $m$ to sum over). We omit details. In some regions the critical exponent differs from those given in \cref{table:PG_asymps}, but the exponential growth rate does not change, and hence $\psi^\textup{GC}(a,y) = \psi^\textup{G}(a,y)$.
\end{proof}

\subsection{Binomial bridges}\label{ssec:bridges}

The lower bound~\eqref{eqn:GC_vs_C} is valid in region $G_\textup{III}$ but is not sharp there. To obtain a sharp lower bound on the free energy of centred polygons in region \(G_\textup{III}\), we will use a different construction.

We define a \emph{binomial bridge} (or just \emph{bridge}) of length \(n\) and height \(h\) to be a binomial path \((\omega_0,\omega_1,\dots,\omega_n)\) satisfying 
\[0 = \y(\omega_0) < \y(\omega_i) \leq \y(\omega_n) = h \qquad \text{for } i=1,\dots,n-1.\]
That is, it is a binomial path which starts at \(\y=0\) and ends at \(\y=h\), with all non-terminal vertices having \(\y\)-coordinate between 1 and \(h\). Let \(b_{n,h}\) be the number of bridges of length \(n\) and height \(h\), and define the partition function
\[B_n(y) = \sum_{h=1}^n b_{n,h} y^h.\]
\begin{lem}\label{lem:bridge_asymp}
For \(y>1\),
\begin{equation}\label{eqn:bridge_asymp}
	\lim_{n\to\infty} \frac1n \log B_n(y) = \lambda^\textup{P}(y) = \log\left(y^2+1\right) - \log y.
\end{equation}	
\end{lem}
\noindent This statement is proved in \cref{app:bridges}.

\cref{lem:bridge_asymp} will allow us to prove the following.
\begin{lem}\label{lem:bridge_lb}
For \(y>1\),
\begin{equation}\label{eqn:bridge_lb}
	\liminf_{n\to\infty} \frac{1}{2n} \log P^\textup{C}_{2n}(a,y) \geq \lambda^\textup{P}\left(\sqrt{y}\right) = \log(y+1) - \frac12\log y.
\end{equation}
\end{lem}

\begin{proof}
For given \(n\) and \(y>1\), define \(h^* \equiv h^*(n,y)\) to be the value of \(h\) between 1 and \(n\) which maximises \(b_{n,h}y^h\) (the so-called ``most popular \(h\)''). Then
\[b_{n,h^*}y^{h^*} \leq B_n(y) \leq nb_{n,h^*}y^{h^*},\]
so that
\[\lim_{n\to\infty} \frac1n \log\left(b_{n,h^*}y^{h^*}\right) = \lambda^\textup{P}(y).\]

Now we can construct a centred polygon of length \(4n\) and height \(2h\) by reflecting, translating and joining together four bridges of length \(n\) and height \(h\), as per \cref{fig:four_bridges} (left). In particular, given \(y>1\), this can be done with $h=h^*$. Then
\begin{equation}\label{eqn:bridges_concat_4n}
a\left(b_{n,h^*}y^{h^*}\right)^4 \leq P^\textup{C}_{4n}\left(a,y^2\right),
\end{equation}
where the factor of $a$ appears because the constructed polygon would have a single surface visit.

\begin{figure}
\centering
\begin{subfigure}{0.45\textwidth}
\centering	
\begin{tikzpicture}
\draw[line width=1ex, gray] (0,0) -- (4,0);
\draw[very thick, blue, decorate, decoration=snake] (2,0) -- (0,2);
\draw[very thick, blue, decorate, decoration=snake] (2,0) -- (4,2);
\draw[very thick, blue, decorate, decoration=snake] (0,2) -- (2,4);
\draw[very thick, blue, decorate, decoration=snake] (2,4) -- (4,2);
\node [circle, fill=red, inner sep=2.5pt] at (2,0) {};
\draw[very thick, <->] (5,0) -- (5,4);
\node at (5,2) [right=1ex] {$2h$};
\node at (0,4.4) {};
\end{tikzpicture}
\end{subfigure}
\hfill
\begin{subfigure}{0.45\textwidth}
\centering	
\begin{tikzpicture}
\draw[line width=1ex, gray] (0,0) -- (4.5,0);
\draw[very thick, blue, decorate, decoration=snake] (2,0) -- (0,2);
\draw[very thick, blue, decorate, decoration=snake] (2.5,0.5) -- (4.5,2.5);
\draw[very thick, blue, decorate, decoration=snake] (0,2) -- (2,4);
\draw[very thick, blue, decorate, decoration=snake] (2.5,4.5) -- (4.5,2.5);
\draw[ultra thick, blue] (2,0) -- (2.5,0.5);
\draw[ultra thick, blue] (2,4) -- (2.5,4.5);
\node [circle, fill=red, inner sep=2.5pt] at (2,0) {};
\draw[very thick, <->] (5.5,0) -- (5.5,4);
\node at (5.5,2) [right=1ex] {$2h$};
\end{tikzpicture}
\end{subfigure}
\caption{The construction of centred polygons of length $4n$ (left) and $4n+2$ (right) and height $2h$ from four bridges of length $n$ and height $h$.}
\label{fig:four_bridges}
\end{figure}
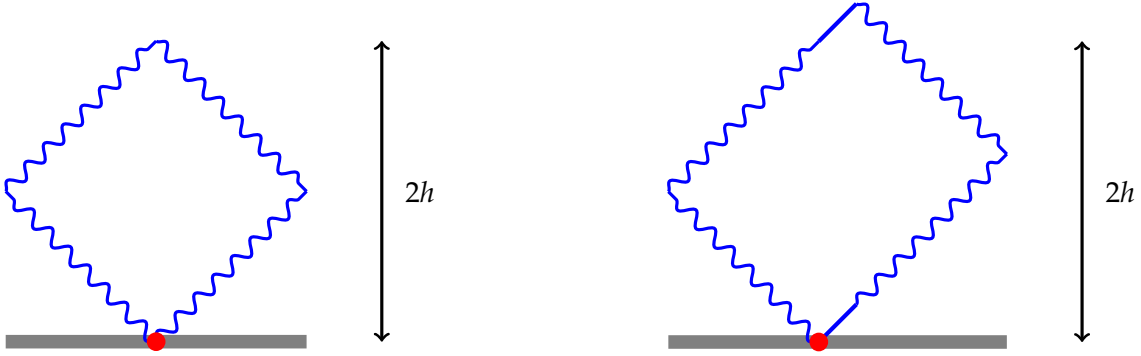

Centred polygons of length $4n+2$ can be constructed in a similar way, with the addition of two extra edges, as per \cref{fig:four_bridges} (right). This gives
\begin{equation}\label{eqn:bridges_concat_4n+2}
a\left(b_{n,h^*}y^{h^*}\right)^4 \leq P^\textup{C}_{4n+2}\left(a,y^2\right).
\end{equation}
Taking logs, dividing by $4n$ and sending $n\to\infty$ in~\eqref{eqn:bridges_concat_4n} and~\eqref{eqn:bridges_concat_4n+2}, we obtain
\[\liminf_{n\to\infty} \frac{1}{2n} \log P^\textup{C}_{2n}\left(a,y^2\right) \geq \lambda^\textup{P}(y),\]
and~\eqref{eqn:bridge_lb} follows.
\end{proof}

By examining the lower bounds given by \cref{lem:GC_comparison,lem:bridge_lb}, it can be	 seen that the RHS of~\eqref{eqn:psiC} is a lower bound for $\liminf_{n\to\infty} \frac{1}{2n} \log P^\textup{C}_{2n}(a,y)$.

\section{Upper bounds via noninteracting pairs of paths}\label{sec:ub}

\subsection{Upper bounds for centred polygons}\label{ssec:centred_ub}

To obtain a tight upper bound on the free energy of centred polygons, we will again consider configurations of pairs of paths. This time, however, we will relax the nonintersecting constraint. This simplifies matters, as the number of pairs of paths with no nonintersecting constraint is just the square of the number of single paths. We must still account for surface interactions, however.

Let $\mathcal B_n$ be the set of binomial paths of length $n$ in $\Lp$ which have leftmost vertex in the line $\x=0$. We then define $t_n(i,k)$ to be the total weight of all paths in $\mathcal B_n$ which start at height $\y=i$ and end at height $\y=k$, and accumulate a weight $a$ for each vertex in the surface $\y=0$. Also define the specialisation $t^1_n(i,k) = t_n(i,k)|_{a=1}$. (Note that $t_n(i,k)$ is nonzero only if $i-k \equiv n\,(\text{mod }2)$.) As with pairs of nonintersecting paths, the quantity $t_n(i,k)$ can be computed using a constant term method. The result is given in~\cite{vanRensburg2000Statistical}, to which we direct the reader for details.

\begin{lem}[{\cite[Eqn.~(5.33)]{vanRensburg2000Statistical}}]
\label{lem:single_path}
\begin{multline}\label{eqn:single_path}
t_n(i,k) = \binom{n}{\frac12(n-i+k)} - \binom{n}{\frac12(n+i+k)} \\ 
+ a\sum_{w=0}^{\frac12(n-i-k)}\left[\binom{n}{\frac12(n+i+k)+w} - \binom{n}{\frac12(n+i+k)+w+1}\right](a-1)^w.
\end{multline}
In particular,
\begin{equation}\label{eqn:single_path_a1}
t_n^1(i,k) = \binom{n}{\frac12(n-i+k)} - \binom{n}{\frac12(n+i+k)+1}.	
\end{equation}
\end{lem}

We now define
\begin{equation}
T_{4n}(a,y) = \frac1a \sum_{k=0}^{2n} \sum_{i_1,i_2=0}^n t_n^1(i_1,k) t_n(i_1,0) t_n^1(k,i_2) t_n(0,i_2) y^k 
= \frac1a \sum_{k=0}^{2n} \left(\sum_{i=0}^n t_n^1(i,k) t_n(i,0)\right)^2 y^k. \label{eqn:Tn_even}
\end{equation}
Observe that the summands count configurations of four (not necessarily nonintersecting) binomial paths: from $(0,i_1)$ to $(n,k)$ (without $a$ weights), from $(0,i_1)$ to $(n,0)$ (with $a$ weights), from $(n,k)$ to $(2n,i_2)$ (without $a$ weights), and from $(n,0)$ to $(2n,i_2)$ (with $a$ weights).

We likewise define
\begin{equation}\label{eqn:Tn_odd}
T_{4n+2}(a,y) = \frac1a \sum_{k=0}^{2n} \left(\sum_{i_1=0}^n t_n^1(i_1,k) t_n(i_1,0)\right) \left(\sum_{i_2=0}^{n+1}  t_{n+1}^1(k,i_2) t_{n+1}(0,i_2)\right) y^k.
\end{equation}
This is the same as~\eqref{eqn:Tn_even}, except the final point is $(2n+2,i_2)$ instead of $(2n,i_2)$.

Note that the objects counted by $T_{2n}(a,y)$ include all centred polygons of length $2n$: these are simply the configurations where the paths do not intersect. We thus have
\begin{equation}\label{eqn:fourpaths_vs_centred}
T_{2n}(a,y) \geq P^\textup{C}_{2n}(a,y)	
\end{equation}
for all $a,y>0$.

Since~\eqref{eqn:fourpaths_vs_centred} is a fairly crude upper bound, we are not interested in the detailed asymptotic behaviour of $T_{2n}(a,y)$. Instead, only the exponential growth rate is of interest. To determine this it is only necessary to find the values of $w,k$ and $i$ (or $i_1$ and $i_2$) which give the greatest contributions to~\eqref{eqn:Tn_even} and~\eqref{eqn:Tn_odd}. These can be found using the same methods as in \cref{sec:grafted_asymps}. (The subsequent summations and integrals used in \cref{sec:grafted_asymps} were only necessary for determining the critical exponents, not the growth rates.) We summarise the results in \cref{table:twopaths_greatest}.

\begin{table}
\small
\centering
\begin{tabular}{| C{0.9cm} | C{4.6cm} | C{4.6cm} | C{4.6cm} |}
\cline{2-4}
\multicolumn{1}{C{0.9cm} |}{} & \btm a<g(y) \etm & \btm a=g(y) \etm & \btm a>g(y) \etm \\ \hline
\btm y>1 \etm & \btm w = O(1),\, i \sim \gamma n,\, k \sim 2\gamma n \etm & \btm w \sim \frac{(1-\epsilon_n)\gamma}{2}n, \, i \sim \epsilon_n\gamma n, \newline k \sim (1+\epsilon_n)\gamma n, \, 0 \leq \epsilon_n \leq 1 \etm & \btm w \sim \delta n,\, i = O(1), \, k \sim \gamma n \etm \\ \hline
\btm y=1 \etm & \btm w = O(1),\, i = O(\sqrt{n}), \newline k = O(\sqrt{n}) \etm & \btm w = O(\sqrt{n}),\, i = O(\sqrt{n}), \newline k = O(\sqrt{n}) \etm & \btm w \sim \delta n, \, i = O(1),\newline k = O(\sqrt{n}) \etm \\ \hline
\btm y<1 \etm & \btm w = O(1),\, i = O(\sqrt{n}), \newline k = O(1) \etm & \btm w = O(\sqrt{n}), \, i=O(\sqrt{n}), \newline k = O(1) \etm & \btm w \sim \delta n, \, i=O(1), \, k = O(1) \etm \\ \hline
\end{tabular}
\caption{The values of $w,i$ and $k$ which give the greatest contributions to $T_{4n}(a,y)$, as per~\eqref{eqn:Tn_even}. Here $\gamma = \frac{y-1}{y+1}$, $\delta = \frac{a-2}{2a}$, and $g(y)$ is the phase boundary given by $g(y) = 2$ if $y \leq 1$ and $g(y) = y+1$ if $y>2$. For the $y>1$, $a=g(y)$ case, $\epsilon_n$ can be any quantity between 0 and 1 (not necessarily a constant), and the growth rate will be independent of $\epsilon_n$. For $T_{4n+2}(a,y)$ as per~\eqref{eqn:Tn_odd}, $i_1$ and $i_2$ take the same values as $i$ given here.}
\label{table:twopaths_greatest}
\end{table}

Once the dominant values of $w, k$ and $i$ (or $i_1$ and $i_2$) have been found, finding the exponential growth rate of $T_{2n}(a,y)$ is simply a matter of applying Stirling's approximation to the summands of~\eqref{eqn:Tn_even} and~\eqref{eqn:Tn_odd}, substituting the values from \cref{table:twopaths_greatest}, taking logs, dividing by $4n$ and taking $n\to\infty$. We omit details. The resulting growth rate is exactly given by the RHS of~\eqref{eqn:psiC}.

\begin{proof}[Proof of \cref{thm:psiC}]
By the results of \cref{sec:centred_lbs}, the RHS of~\eqref{eqn:psiC} has been shown to be a lower bound for $\liminf_{n\to\infty} \frac{1}{2n} \log P^\textup{C}_{2n}(a,y)$. By the results of this section, the RHS of~\eqref{eqn:psiC} is also the growth rate of $T_{2n}(a,y)$. Then by~\eqref{eqn:fourpaths_vs_centred}, this is an upper bound for $\limsup_{n\to\infty} \frac{1}{2n} \log P^\textup{C}_{2n}(a,y)$. The theorem immediately follows.
\end{proof}

Similarly to grafted polygons, the regions $C_\textup{I}$, $C_\textup{II}$, $C_\textup{III}$ and $C_\textup{IV}$ correspond to the four different phases for centred polygons: free, adsorbed, ballistic, and mixed (adsorbed and ballistic) respectively. See \cref{fig:phase_diagrams}.

\subsection{Upper bounds for all half-space polygons}\label{ssec:ub_all}

By slightly generalising the arguments of the previous section, we can obtain in a similar way a sharp upper bound for $\psi^\textup{S}(a,y)$, the free energy of all half-space polygons. Fix $q$ between $0$ and $2n$, and consider a pair of binomial paths, both starting at $(0,i_1)$ and finishing at $(2n,i_2)$ (but not necessarily nonintersecting). Designate one path as the ``lower path'', and only consider cases where the lower path has at least one vertex in the surface $\y=0$. Furthermore, consider only cases where the leftmost vertex of the lower path which is in the surface occurs at $(q,0)$. Let the lower path gain a weight $a$ for each vertex in the surface, and associate a weight $y$ with the height of the middle vertex of the upper path, as usual. Let $S_{2n,q}(a,y)$ be the total weight of all such configurations (i.e.~sum over all $i_1$ and $i_2$).

Using $t_n(i,k)$ and $t_n^1(i,k)$ as defined in the previous section, it is not difficult to see that if $q>1$ we have
\begin{equation}\label{eqn:S4n_q}
S_{4n,q}(a,y) = \sum_{k=0}^{\min\{n+q,3n-q\}}\sum_{i_1=1}^q\sum_{i_2=0}^{2n-q}t^1_n(i_1,k)t^1_{q-1}(i_1-1,0)t^1_n(k,i_2)t_{2n-q}(0,i_2)y^k,
\end{equation}
while for $q=0$
\begin{equation}
S_{4n,0}(a,y) = \sum_{k=0}^n \sum_{i_2=0}^{2n} t^1_n(0,k)t^1_n(k,i_2)t_{2n}(0,i_2)y^k.	
\end{equation}
Similar expressions can be found for $S_{4n+2,q}(a,y)$.

Then define
\begin{equation}\label{eqn:S4n_sum}
S_{4n}(a,y) = \sum_{q=0}^{2n} S_{4n,q}(a,y)
\end{equation}
and similarly for $S_{4n+2}(a,y)$. The relationship between $S_{2n}(a,y)$ and $P^\textup{S}_{2n}(a,y)$ is the same as that of $T_{2n}(a,y)$ and $P^\textup{C}_{2n}(a,y)$: the polygons are exactly the configurations where the upper and lower paths do not intersect. We thus have
\begin{equation}
S_{2n}(a,y) \geq P^\textup{S}_{2n}(a,y)	
\end{equation}
for all $a,y>0$.

We then proceed as in the previous section. We are interested only in the asymptotic growth rate of $S_{2n}(a,y)$, and thus in the values of $w,k,i_1,i_2$ and $q$ which contribute the most to $S_{2n}(a,y)$. The relevant values of $w,k,i_1$ and $i_2$ are as given in \cref{table:twopaths_greatest}. For $q$, a similar analysis can be performed, with the dominant values given in \cref{table:twopaths_q_greatest}. Note that in some cases the growth rate is entirely independent of $q$.

With the relevant values of $q$ in hand, the asymptotics of $S_{2n}(a,y)$ can be found in the same way as for $T_{2n}(a,y)$: apply Stirling's approximation to the summands of~\eqref{eqn:S4n_q}, substitute the relevant values of $w,i_1,i_2,k$ and $q$, take logs, divide by $2n$ and send $n\to\infty$. We omit details. The resulting growth rate is found to be equal to the RHS of~\eqref{eqn:psiC}.

\begin{proof}[Proof of \cref{thm:psi+}]
The proof is analogous to that of \cref{thm:psiC}. We have 
\[P^\textup{C}_{2n}(a,y) \leq P^\textup{S}_{2n}(a,y) \leq S_{2n}(a,y).\]
Take logs, divide by $2n$ and send $n\to\infty$. The LHS and RHS both go to the RHS of~\eqref{eqn:psiC}, and the result follows.
\end{proof}

\begin{table}
\centering
\begin{tabular}{| C{1.2cm} | C{3.5cm} | C{3.5cm} | C{3.5cm} |}
\cline{2-4}
\multicolumn{1}{C{1.2cm} |}{} & \btm a<g(y) \etm & \btm a=g(y) \etm & \btm a>g(y) \etm \\ \hline
\btm y>1 \etm & \btm q\sim n \etm & \btm q \sim \epsilon_n n, \, 0 \leq \epsilon_n \leq 1 \etm & \btm q = O(1) \etm \\ \hline
\btm y=1 \etm & \btm * \etm & \btm * \etm & \btm q = O(1) \etm \\ \hline
\btm y<1 \etm & \btm * \etm & \btm * \etm & \btm q = O(1) \etm \\ \hline
\end{tabular}
\caption{The value of $q$ which gives the greatest contribution to $S_{4n}(a,y)$, as per~\mbox{\eqref{eqn:S4n_q}--\eqref{eqn:S4n_sum}}. For the $y>1$, $a=g(y)$ case, $\epsilon_n$ is as per \cref{table:twopaths_greatest}. For the four cases with $*$, the growth rate is entirely independent of $q$. The results are the same for $S_{4n+2}(a,y)$.}
\label{table:twopaths_q_greatest}
\end{table}

\begin{figure}
\centering
\begin{subfigure}{0.49\textwidth}
\centering
\resizebox{\textwidth}{!}{
\begin{tikzpicture}[scale=1.3]
\draw[thin, <->] (0,4) -- (0,0) -- (5,0);
\draw[very thick, blue, ->] (2,0) -- (2,4);
\draw[very thick, blue, ->] (0,1) -- (5,1);
\node at (1,0.5) {$G_\textup{I}$ (free)};
\node at (3.5,0.5) {$G_\textup{II}$ (adsorbed)};
\node at (1,2.5) {$G_\textup{III}$ (ballistic)};
\node at (3.5,2.5) {$G_\textup{IV}$ (mixed)};
\node at (5,0) [right] {$a$};
\node at (0,4) [above] {$y$};
\node at (2,0) [below] {$a_\text{c} = 2$};
\node at (0,1) [left] {$y_\text{c} = 1$};
\end{tikzpicture}
}
\caption{}
\label{fig:phase_diag_grafted}
\end{subfigure}
\begin{subfigure}{0.49\textwidth}
\centering
\resizebox{\textwidth}{!}{
\begin{tikzpicture}[scale=1.3]
\draw[thin, <->] (0,4) -- (0,0) -- (5,0);
\draw[very thick, blue] (2,0) -- (2,1);
\draw[very thick, blue, ->] (0,1) -- (5,1);
\draw[very thick, blue, dashed, ->] (2,1) -- (5,4);
\node at (1,0.5) {$C_\textup{I}$ (free)};
\node at (3.5,0.5) {$C_\textup{II}$ (adsorbed)};
\node at (1.5,2.5) {$C_\textup{III}$ (ballistic)};
\node at (4.2,1.8) {$C_\textup{IV}$ (mixed)};
\node at (5,0) [right] {$a$};
\node at (0,4) [above] {$y$};
\node at (2,0) [below] {$a_\text{c} = 2$};
\node at (0,1) [left] {$y_\text{c} = 1$};
\node at (4,3) [above left, rotate=45] {$y=a-1$};
\end{tikzpicture}
}
\caption{}
\label{fig:phase_diag_grafted}
\end{subfigure}
\caption{\textbf{(a)} The phase diagram for grafted staircase polygons, as implied by \cref{thm:psiG}. \textbf{(b)} The phase diagram for centred staircase polygons, and thus all staircase polygons, as implied by \cref{thm:psiC,thm:psi+}. Solid lines indicate second-order phase transitions, while dashed lines indicate first-order transitions.}
\label{fig:phase_diagrams}
\end{figure}
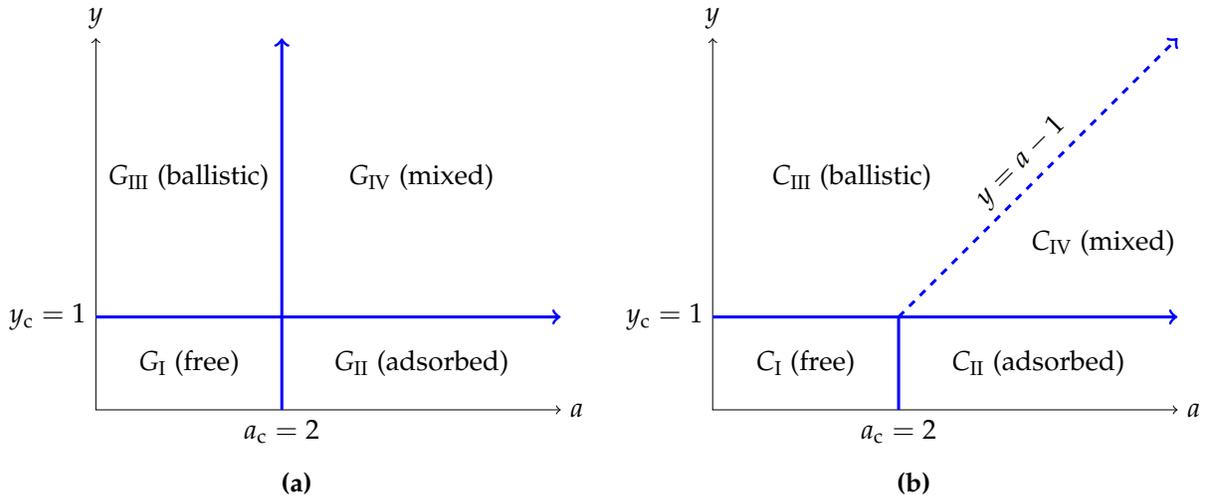

\section{Conclusion}\label{sec:conclusion}

We have defined several models of staircase polygons which interact with an impenetrable surface and have a pulling or pushing force applied at a fixed vertex. These are greatly simplified versions of two-dimensional self-avoiding polygons, which nevertheless allow for both exact solvability and a range of different physical behaviours.

The phase diagrams for the models solved here are given in \cref{fig:phase_diagrams}. The locations and natures of the phase transitions follow immediately from the free energies given in \cref{thm:psiG,thm:psiC}.

There are still open questions regarding the more general SAP model. Perhaps the most important question is whether two-dimensional SAPs exhibit a ``mixed'' phase, as staircase polygons do in region $G_\textup{IV}$. (It is known~\cite{GvRJW2017Polygons} that such behaviour does not occur for SAPs in three or more dimensions.) If SAPs do have a mixed phase then another question arises: is the phase boundary between the adsorbed and mixed phases a horizontal line at $y=1$ (as it is for staircase polygons), or does it increase with $a$? For staircase polygons the horizontal bounday is inevitable (because we cannot move from the mixed phase to the adsorbed phase by increasing $a$), but for a more general model this is not clear.

There are other ``intermediate'' models which are still simpler than SAPs but may display more complex asymptotic behaviour, the most obvious being column-convex polygons on the square lattice. Exact solutions for such models may not be easy to come by, however. Another alternative is to consider staircase polygons with a force applied at the highest vertex, rather than the middle vertex. As discussed in~\cite{GvRJW2017Polygons}, one would expect the free energy to be the same in all regions of the phase diagram, but when $y<1$ the sub-dominant asymptotic behaviour of the partition functions will likely involve more complex ``stretched exponential'' terms (see~\cite{BGJL_2015,Guttmann_2015}).

\section*{Acknowledgements}

The author is supported by the Australian Research Council project DE170100186. This work began while the author was visiting the Department of Chemistry at the University of Toronto, and their hospitality was greatly appreciated. The author thanks Stu Whittington, Buks van Rensburg and Tony Guttmann for helpful conversations.

\appendix

\section{Binomial bridges}\label{app:bridges}

This section is dedicated to a proof of \cref{lem:bridge_asymp}. Binomial bridges (sometimes called \emph{directed bridges}) have been studied in the past -- see for example~\cite{Brak_1999,Brak2005Directed}. However, we have been unable to find in the literature the exact result that we require, and thus we present a brief proof here.

First, recall the generating function $C(t;a,y)$ for half-plane binomial paths from \cref{ssec:binomial_paths}. We will not need surface visits here, so we set $a=1$. We will then similarly define the generating function for bridges as
\[B(t;y) = \sum_{n,h} b_{n,h} t^n y^h.\]

Now there is a standard method for decomposing a half-plane walk (of the type counted by $C(t;1,y)$) into an alternating sequence of bridges and reflected bridges. We refer the reader to~\cite[Chap.~3]{Madras1993SelfAvoiding} for details (the focus there is on self-avoiding walks, but the same principles apply). In particular, the equivalent form of~\cite[Eqn.~(3.1.13)]{Madras1993SelfAvoiding} here is
\begin{equation}\label{eqn:halfplane_vs_bridges_no_y}
tC(t;1,1) \leq \exp\left(B(t;1)\right).
\end{equation}
To incorporate the weight $y$, we use the fact that appending a bridge to a walk will increase the height of its endpoint, while attaching a reflected bridge will decrease the height. The inequality~\eqref{eqn:halfplane_vs_bridges_no_y} becomes
\begin{equation}\label{eqn:halfplane_vs_bridges_y}
tyC(t;1,y) \leq \exp\left(B(t;y) + B(t;1/y)\right).
\end{equation}

\begin{proof}[Proof of \cref{lem:bridge_asymp}]
Since the exponential function is entire, the radius of convergence of the RHS of~\eqref{eqn:halfplane_vs_bridges_y} (considered as a power series in $t$) is equal to the radius of convergence of $B(t;y) + B(t;1/y)$. For $y>1$, we have $B_n(y) > B_n(1/y)$, and so that radius of convergence is really that of $B(t;y)$.

Meanwhile, by inclusion we clearly have $B(t;y) \leq tyC(t;1,y)$. Combining these ideas, it follows that for $y>1$ the radius of convergence of $tyC(t;1,y)$ (and hence that of $C(t;y)$) is equal to that of $B(t;y)$. So the growth rate of $C_n(1,y)$ is equal to that of $B_n(y)$, and then by~\eqref{eqn:psiP_maxform}, the result follows.
\end{proof}

\printbibliography

\end{document}